\documentclass[journal]{IEEEtran}

\ifCLASSINFOpdf

\else

\fi

\usepackage{amsmath}
\usepackage{amssymb}
\usepackage{multirow} 
\usepackage{amsthm}
\usepackage{graphicx}
\usepackage[numbers]{natbib} 
\usepackage{booktabs}
\usepackage{array}
\usepackage{makecell}
\usepackage{thmtools}
\usepackage{url}
\newtheorem{theorem}{Theorem}
\newtheorem{definition}{Definition}
\newtheorem{corollary}{Corollary}
\newtheorem{lemma}{Lemma}

\begin{document}

\title{Tur\'{a}n-Theoretic Bounds on Several Elementary Trapping Sets in LDPC Codes}

\author{
\vspace{5mm}
\IEEEauthorblockN{
Ziyang Zhao\IEEEauthorrefmark{2}\IEEEauthorrefmark{3},
Haoran Xiong\IEEEauthorrefmark{3}\IEEEauthorrefmark{4},
Zicheng Ye\IEEEauthorrefmark{3}\IEEEauthorrefmark{4}, 
Luyi Li\IEEEauthorrefmark{3}\IEEEauthorrefmark{4}, 
and
Guiying Yan\IEEEauthorrefmark{2}\IEEEauthorrefmark{3}\IEEEauthorrefmark{4}}

\IEEEauthorblockA{\IEEEauthorrefmark{2}School of 
Advanced Interdisciplinary Sciences, University of Chinese 
Academy of Sciences}

\IEEEauthorblockA{\IEEEauthorrefmark{3}Academy of Mathematics and Systems Science, CAS}

\IEEEauthorblockA{\IEEEauthorrefmark{4}School of Mathematical Sciences, University of Chinese 
Academy of Sciences}

\IEEEauthorblockA{Email: zhaoziyang251@mails.ucas.ac.cn,  \{xionghaoran, yezicheng\}@amss.ac.cn, liluyiplus@gmail.com, yangy@amss.ac.cn}

\IEEEauthorblockA{Corresponding Author: Guiying Yan \quad Email:yangy@amss.ac.cn}

\thanks{This work was supported by the National Key Research and Development Program of China under Grant 2023YFA1009602.}
}

 \maketitle

\begin{abstract}
LDPC codes have attracted significant attention due to their capacity-approaching performance. Elementary trapping sets are the main cause of the error floor phenomenon in LDPC codes. We investigate several graph structures associated with trapping sets, including theta graphs, dumbbell graphs, and short cycles with chords. Based on the Tur\'{a}n numbers of $\theta(2,2,2)$, $\theta(1,3,3)$ and $D(4,4;0)$, we prove that any $(a,b)$-ETS in a variable-regular Tanner graph with girth $g=8$ and variable degree $\gamma$ satisfies the inequality $b\geq a\gamma-\frac{a(\sqrt{24a-23}-1)}{4}$, provided that any two 8-cycles in the Tanner graph do not share common variable node. In addition, we can also eliminate ETSs by removing certain short-cycle structures with chords. The lower bounds on the minimum size of ETSs through these methods are improved. To assess practical impact, we analyze spectral radii of the ETSs and construct QC-LDPC codes to show frame error rates in the error floor region.
\end{abstract}


\IEEEpeerreviewmaketitle

\section{Introduction}

Low-density parity-check (LDPC) codes are capacity-approaching error-correcting codes widely used in modern communication and storage systems, including Wi-Fi, optical/microwave links, and 5G enhanced mobile broadband (eMBB) channels \cite{GallagerLDPC},\cite{ChunLDPC}. However, under iterative decoding, LDPC codes suffer from the error floor phenomenon\cite{ChilaER}. This behavior is mainly caused by certain harmful substructures in the Tanner graph, known as trapping sets, among which elementary trapping sets(ETSs) are regarded as the most detrimental \cite{cole2006},\cite{Milenkovic}.

Considerable effort has been devoted to eliminating trapping sets in Tanner graphs. As approximating the minimum trapping-set size is NP-hard \cite{McGregor}, exhaustive enumeration lacks reliable worst-case guarantees. Consequently, structural constraints on Tanner graphs are widely adopted to exclude small ETSs. Restrictions on short cycles—such as 8-cycles in girth-8 quasi-cyclic LDPC codes \cite{NaseriSiAm,TaoXiongfe} or short cycles with chords \cite{AmirzadePho, NaseriBaniha}—have proven effective in improving minimum distance and error floor performance, with further bounds derived using Tur\'{a}n numbers of theta graphs \cite{XiongYeYan} and variable-regular degree constraints \cite{HashemAmir}. Edge coloring and extremal graph techniques have also been applied to trapping set analysis and code construction \cite{Amirzade, Sadeghi, Kari}. These developments motivate the application of extremal graph theory to derive new theoretical bounds and practical design criteria for LDPC codes.

The structure of this paper is as follows: Section II introduces the essential definitions and notations required for our analysis. In Section III, we determine the Tur\'{a}n numbers for several special graphs, including theta graphs, dumbbell graphs, and short cycles with chords. Section IV ties the mathematical results from Section III to coding theory, enabling us to establish inequalities for the parameters $a$, $b$, and $\gamma$ associated with an $(a, b)$-ETS in a Tanner graph with a variable-regular degree of $d_L(v)=\gamma$. This mathematical foundation allows us theoretically to show the effect of eliminating various small ETSs. In Section V, we exhibit construction examples that employ our proposed methods and the corresponding numerical results. Section VI concludes this paper.

\section{Preliminaries}

\subsection{Graph theory}

A graph $G=(V,E)$ consists of a non-empty vertex set $V$ and an edge set $E$ as unordered pairs of vertices. Two vertices $u$ and $v$ are adjacent, denoted as $uv\in E$ if they are connected by an edge. The neighborhood of a vertex $v$ is defined as $N_G(v) = \{ u \in V \mid uv \in E \}$ which is the set of vertices adjacent to $v$. Similarly, for a subset $S \subseteq V$, its neighborhood is denoted by $N(S)$. The subgraph induced by $S$ is defined as $G[S] = (S, E')$, where $E' = \{ vw \in E \mid v, w \in S \}.$ A graph without loops or multiple edges is simple. 
A simple graph is complete, denoted $K_n$, if there is an edge connecting every pair of its $n$ vertices. Additionally, a simple graph is bipartite if its vertices can be partitioned into two sets $V_1$, $V_2$ with no edge within the same set. A bipartite graph is complete, denoted by $K_{m,n}$, if every vertex in $V_1$ is connected to every vertex in $V_2$, with $m$ and $n$ vertices in each part respectively. In particular, the complete bipartite graph $K_{1,n-1}$ is called an $n$-vertex star graph, denoted by $S_n$.

The degree of a vertex $u$ is defined as $d_G(u) = |N_G(u)|$, representing the number of vertices adjacent to $u$. The complement of $G$ is defined as $\overline{G} = (V, \overline{E})$, where $\overline{E}$ is the complement of the edge set $E$. For an $n$-vertex graph $G$ with vertex set $V = \{ v_1, v_2, \dots, v_n \}$, the degree sequence is denoted by $\pi_G = (d_1\leq d_2\leq \dots\leq d_n)$ where $d_i =d_G(v_i)$. Here, $d_1$ and $d_n$ are the minimum degree and maximum degree of $G$, respectively.

The graph $P_{k} = (V, E)$ with $V = \{x_0, x_1, x_2, \cdots, x_k\}$ and $E = \{x_0 x_1, x_1 x_2, \cdots, x_{k-1} x_k\}$, where all vertices $x_i$ are distinct, is called a path of length $k$. The vertices $x_0$ and $x_k$ are the endpoints of the path $P$. We say that two or more paths are internally disjoint if they have no common vertices except the endpoints. The distance $d(u, v)$ between $u$ and $v$ is defined as the length of a shortest path connecting $u$ and $v$. The diameter of a graph $G$ is defined as $\operatorname{diam}(G) = \max \{d(u, v) \mid u, v \in V(G)\}$. A cycle of length $k$, denoted by $C_k$, is a path that starts and ends at the same vertex $x_0 = x_k$ with all other vertices in the sequence being distinct. The girth of a simple graph is the length of its shortest cycle, denoted by $g$. The following definitions introduce theta graphs\cite{theta} and dumbbell graphs, which are the main graphs considered in this paper.

\begin{definition}
    A theta graph, denoted by $\theta(a, b, c)$, is formed by three internally disjoint paths with the same pair of endpoints, of lengths $a$, $b$, $c$ respectively, where $1 \leq a \leq b \leq c$.
\end{definition}

\begin{definition}
    A dumbbell graph, denoted by $D(a, b; q)$, is formed by two vertex-disjoint cycles $C_a$ and $C_b$ connected by a path of length $q$. In particular, when $q = 0$, $D(a, b; 0)$ indicates that $C_a$ and $C_b$ share one common vertex.
\end{definition}

Two graphs are isomorphic if there exists a bijection between their vertex sets that preserves adjacency and non-adjacency. If a graph $G$ does not contain a subgraph isomorphic to $H$, then $G$ is said to be $H$-free. We now introduce the concept of the Tur\'{a}n number \cite{turan}.

\begin{definition}
     Let $\mathcal{F}$ be a non-empty family of graphs. The Tur\'{a}n number $ex(n, \mathcal{F})$ is defined as the maximum number of edges in an $n$-vertex $\mathcal{F}$-free graph. An $n$-vertex graph attaining this maximum number of edges is called an $\mathcal{F}$-free extremal graph of order $n$. In particular, when $\mathcal{F}$ consists of a single graph $G$, we simply write $ex(n, G)$.
\end{definition}

For two graphs $G$ and $H$, we write $G = H$ to mean that they are isomorphic. The disjoint union and join of $G$ and $H$ are denoted by $G+H$ and $G\lor H$, respectively. The notation $mH$ stands for the disjoint union of $m$ copies of $H$. For a graph $G = (V, E)$ and a vertex $v \in V$, $G - v$ denotes the graph obtained from $G$ by deleting the vertex $v$ together with all edges incident to $v$. 

\subsection{LDPC code}

A Tanner graph $G = (L \cup R, E)$ is a bipartite graph derived from the parity-check matrix $H$ of an LDPC code, which is introduced in \cite{Tanner1}.  In this graph, $L$ and $R$ denote the sets of variable nodes and check nodes, corresponding to the columns and rows of $H$, respectively. The $i$-th check node is adjacent to the $j$-th variable node if and only if $H(i,j) = 1$. The Tanner graph is variable-regular if every variable node $v\in L$ has the same degree $d_L(v)=\gamma$.

For a lifting degree $p$, the parity-check matrix $H$ of a QC-LDPC code can be represented by $p \times p$ circulant permutation matrices as follows:
$$
H=
\begin{bmatrix}
I(p_{1,1}) & I(p_{1,2}) & I(p_{1,3}) & \cdots & I(p_{1,\eta}) \\
I(p_{2,1}) & I(p_{2,2}) & I(p_{2,3}) & \cdots & I(p_{2,\eta}) \\
\vdots & \vdots & \vdots & \ddots & \vdots \\
I(p_{\gamma,1}) & I(p_{\gamma,2}) & I(p_{\gamma,3}) & \cdots & I(p_{\gamma,\eta})
\end{bmatrix},
$$
where $I(p_{i,j})$ represents a $p \times p$ circulant permutation matrix with lifting value $p_{i,j} \in \{0,1,2,\dots,p-1,\infty\}$. For $0 \leq r \leq p-1$, $I(p_{i,j})$ has a $1$ at position $(r, (r+p_{i,j}) \bmod p)$ and $0$ elsewhere. In particular, $I(\infty)$ represents a $p \times p$ zero matrix.

A QC-LDPC code is fully connected if and only if its parity-check matrix $H$ contains no $I(\infty)$ entry. The $\gamma \times \eta$ matrix formed by the lifting values is called the exponent matrix $E$, where $E(i,j)=p_{i,j}$. Furthermore, we define the base matrix $B$ with $B(i,j)=1$ if $p_{i,j}\neq \infty$, and $B(i,j)=0$ otherwise.

The definitions of trapping set and VN graph\cite{KoetterVN} are given below.

\begin{definition}
    An $(a,b)$ trapping set is an induced subgraph $G[S \cup N(S)]$ generated by $S \subseteq L$ and $N(S) \subseteq R$ where $a=|S|$ represents the size and $b$ is the number of vertices in $N(S)$ with odd degree. Furthermore, if all check nodes have degrees of either $1$ or $2$, it is called an elementary trapping set (ETS).
\end{definition}

\begin{definition}
    For a given ETS, define the variable-node (VN) graph as $G_{\mathrm{VN}} = (V_{\mathrm{VN}}, E_{\mathrm{VN}})$. The vertex set $V_{\mathrm{VN}}$ consists of all variable nodes, and the edge set $E_{\mathrm{VN}}$ consists of the pairs of variable nodes that are connected by a check node with degree 2.
\end{definition}

For a variable-regular LDPC code where every variable node $v \in L$ has the same degree $d_L(v) = \gamma$, there is a correspondence between all elementary trapping sets and their VN graphs. Let $G_{\mathrm{VN}} = (V_{\mathrm{VN}}, E_{\mathrm{VN}})$ be the VN graph of an $(a, b)$-ETS. It follows that $|V_{\mathrm{VN}}| = a$ and $|E_{\mathrm{VN}}| = \frac{1}{2}(a\gamma - b)$. When the Tanner graph has girth $g > 4$, the VN graph of ETS contains neither multiple edges nor loops, satisfying the conditions of a simple graph.

To analyze the decoding behavior within an ETS, \citet{Butler1} introduced a linear state-space model. The system matrix associated with an ETS is defined as follows:

\begin{definition} 
For a given $(a,b)$-ETS, let $G$ be its corresponding VN graph. Define $D(G)=\{(u,v),(v,u) \mid uv \in E(G)\}$. The system matrix $A_{sys}$ is a $|D(G)| \times |D(G)|$ matrix whose rows and columns correspond to ordered pairs in $D(G)$. $A_{sys}(G)(i,j) = 1$ if and only if the second component of ordered pair $j$ equals the first component of ordered pair $i$, and the first component of $j$ is not equal to the second component of $i$; otherwise, $A_{sys}(G)(i,j) = 0$.
\end{definition}

The system matrix $A_{sys}$ characterizes the main message-passing process within the ETS, and its spectral radius $\rho(A_{sys})$ is an important parameter influencing the decoding behavior. \citet{Butler1} demonstrated that when errors occur in an ETS, a larger spectral radius leads to slower error reduction and may even result in decoding failure. Therefore, we regard the spectral radius of $A_{sys}$ as a key metric for assessing the harmfulness of an ETS to the error-floor performance of LDPC codes: a larger spectral radius indicates a more detrimental ETS.

\section{Theoretical Result}

This section focuses on Tur\'{a}n numbers for some special graphs.

\begin{figure}[!t]
\centering
\includegraphics[width=3.2in]{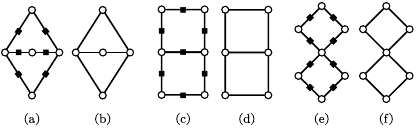}
\caption{Figures (a) and (b) depict two 8-cycles that share two common check nodes and the corresponding VN graph $\theta(2, 2, 2)$, respectively. Figures (c) and (d) depict two 8-cycles with a shared check node and the corresponding VN graph $\theta(1, 3, 3)$, respectively. Figures (e) and (f) depict two 8-cycles with a shared variable node and the corresponding VN graph $D(4,4;0)$, respectively. }
\label{fig1}
\end{figure}

\subsection{Tur\' {a}n Numbers of Theta Graph and Dumbbell Graph}
Given the hardness of determining the exact Tur\'{a}n number for bipartite graphs, we focus on deriving an upper bound for the Tur\'{a}n numbers of $\theta(2,2,2)$ and $\{C_3, \theta(2,2,2), \theta(1,3,3), D(4,4;0)\}$.

Let $ \overline{d} = \frac{1}{n} \sum_{v \in V} d_G(v) $ denote the average degree, $ \sigma^2 = \sum_{v \in V} [\overline{d} - d_G(v)]^2 $ the variance of the degree, and $ N_i $ the number of unordered pairs of vertices at distance $ i $. If $\mathrm{diam}(G)$ satisfies $\mathrm{diam}(G) \leq j$, then $\binom{n}{2} = \sum_{i=1}^j N_i$.

\begin{theorem}\label{th1}
    For all $n \geq 5$, $ex(n,\theta(2,2,2))\leq\frac{n(\sqrt{8n-7}+1)}{4}$. 
\end{theorem}
\begin{IEEEproof}
    Let $G=(V,E)$ be an extremal graph with $|V|=n$ and $|E|=m$, containing no $\theta(2,2,2)$. We claim that the diameter of $G$ is at most 3. Otherwise, there exists a pair of vertices $(x,y)$ such that $d(x,y) \geq 4$. Adding the edge $xy$ to $E(G)$ would not create a new $C_4$, hence $G+xy$ would remain $\theta(2,2,2)$ free, which is a contradiction since $G$ is assumed to be an extremal graph. 
    
    Given the diameter is at most 3, it follows that $\binom{n}{2} = N_1 + N_2 + N_3.$ Without $\theta(2,2,2)$, any common neighbors of a pair with distance 1 or 2 are at most two. Then, we obtain $2(N_1 + N_2) \geq \sum_{v \in V} \binom{d_G(v)}{2}$.
    
     Using the notations introduced earlier, we obtain $\sum_{v \in V} \binom{d_G(v)}{2} = \frac{1}{2} ( \sum_{v \in V} d_G(v)^2 ) - m$ and $\sigma_G^2 = \sum_{v \in V} d_G(v)^2 - \frac{4m^2}{n}$. Combining the above equalities, we further derive $\binom{n}{2} = N_1 + N_2 + N_3 \geq \frac{\sigma_G^2}{4} + \frac{m^2}{n} - \frac{m}{2} \geq \frac{m^2}{n} - \frac{m}{2}$, which implies that $m \leq \frac{n(\sqrt{8n-7}+1)}{4}. $
\end{IEEEproof}

\begin{theorem}\label{th2}
    For all $n \geq 5$ and $\mathcal{H}=\{C_3, \theta(2,2,2), \theta(1,3,3), D(4,4;0)\}$, 
    $$ex(n,\mathcal{H})\leq\frac{n(\sqrt{24n-23}-1)}{8}.$$
\end{theorem} 
\begin{IEEEproof}
    Let $G=(V,E)$ be an extremal graph with $|V|=n$ and $|E|=m$, containing none of $\mathcal{H}$. We claim that the diameter of $G$ is at most 3. Otherwise, there exists a pair of vertices $(x,y)$ such that $d(x,y) \geq 4$. Adding the edge $xy$ to $E(G)$ would not create a new $C_3$ and $C_4$, hence $G+xy$ would remain $\mathcal{H}$ free, which is a contradiction since $G$ is assumed to be an extremal graph. 
    
    Given the diameter is at most 3, it follows that $\binom{n}{2} = N_1 + N_2 + N_3$. Let $Q = \{ \{x,y\} \mid |N(x) \cap N(y)| = 2 \}$ denote the set of vertex pairs sharing exactly two common neighbors. As the graph is $ \theta(2,2,2)$-free, we have $\sum_{v \in V} \binom{d_G(v)}{2} = N_2 + |Q|$. Without $\theta(1,3,3)$ and $D(4,4;0)$, we obtain $|Q| \leq \frac{N_2}{2}$. Furthermore, we conclude that $\sum_{v \in V} \binom{d_G(v)}{2} \leq \frac{3N_2}{2}$. Using the notations introduced earlier, we obtain $\sum_{v \in V} \binom{d_G(v)}{2} = \frac{1}{2} ( \sum_{v \in V} d_G(v)^2 ) - m$ and $\sigma_G^2 = \sum_{v \in V} \bar{d}^2 - 2 \sum_{v \in V} \bar{d} d_G(v) + \sum_{v \in V} d_G(v)^2 = \sum_{v \in V} d_G(v)^2 - \frac{4m^2}{n}$.  Combining the above equalities, we further derive $\binom{n}{2} = N_1 + N_2 + N_3 \geq m + \frac{\sigma_G^2}{3} + \frac{4m^2}{3n} - \frac{2m}{3} \geq \frac{4m^2}{3n} + \frac{m}{3}$, which implies that $m \leq \frac{n(\sqrt{24n-23}-1)}{8}$. 
\end{IEEEproof}

\subsection{Tur\'{a}n numbers of Short Cycles with Chords}

\begin{figure}[!t]
\centering
\includegraphics[width=2.9in]{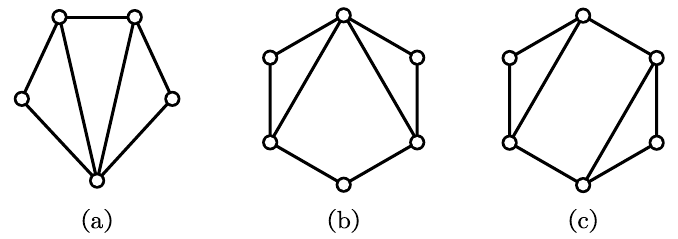}
\caption{Figures (a), (b) and (c) depict three short cycles with chords $G_1$, $G_2$ and $G_3$, respectively.}
\label{fig2}
\end{figure}

This subsection primarily presents the exact results on the Tur\'{a}n numbers for the three short cycles with chords $G_1$, $G_2$, and $G_3$. Additionally, we provide specific characterizations of the extremal graphs for $G_1$ and $G_3$.

In order to further determine the Tur\'{a}n numbers for some special graphs, We first prove the following lemma: 

\begin{lemma}\label{lem1}
    Let $n\geq3$ be a positive integer and $f: \mathbb{N} \to \mathbb{N}$ a function. Define  
$$
F(n) = (n-1)\bigl(f(n+1)+1\bigr) - (n+1)f(n).
$$  

For a given graph $H$, if $ex(n, H) = f(n)$, then the following statements hold:

\begin{enumerate}
    \item If $F(n)>0$, then $ex(n+1, H) \leq f(n+1)$.
    \item If $F(n)=0$ or -1, then $ex(n+1, H)\leq f(n+1)+1$. Moreover, the equality holds if and only if there exists an $n+1$-vertex $H$-free graph $G$ with minimum degree $f(n+1)-f(n)+1$.
    \item If $F(n)\geq -1$ and $\operatorname{ex}(n+1,H)=f(n+1)$, then every
$(n+1)$-vertex extremal $H$-free graph $G$ satisfies $\delta(G)=f(n+1)-f(n)$.
\end{enumerate}
\end{lemma}
\begin{IEEEproof}
    We first prove 1). Suppose that a graph $G$ has $n+1$ vertices and $f(n+1)+1$ edges. We show that $G$ must contain a subgraph isomorphic to $H$. First, if every vertex $v$ in $G$ satisfies $d(v) \geq f(n+1) - f(n) + 1$, then the number of edges in $G$ satisfies
    $
    |E(G)| = \frac{1}{2} \sum_{v \in V(G)} d(v) \geq \frac{1}{2} (n+1) \bigl(f(n+1) - f(n) + 1\bigr). 
    $
    On the other hand, we have $(n-1)(f(n+1)+1) - (n+1)f(n) > 0$, which implies $f(n) < \frac{n-1}{n+1}(f(n+1)+1)$. Consequently, $|E (G)| \geq \frac{n+1}{2}(f(n+1)-f(n)+1)>f(n+1) + 1$, contradicting the assumption. Therefore, there must exist a vertex $v_0$ in $G$ whose degree satisfies $d(v_0) \leq f(n+1) - f(n)$. Consider the graph $G - v_0$, which has $|V(G-v_0)| = n$ vertices and $|E(G-v_0)| \geq f(n) + 1$ edges. Since $ex(n, H) = f(n)$, the graph $G-v_0$ contains a subgraph isomorphic to $H$, and hence so does $G$.

    For 2), suppose that a graph $G$ has $n+1$ vertices and $f(n+1)+2$ edges. We show that $G$ must contain a subgraph isomorphic to $H$. First, if every vertex $v$ in $G$ satisfies $d(v) \geq f(n+1) - f(n) + 2$, then the number of edges in $G$ satisfies $|E(G)| = \frac{1}{2} \sum_{v \in V(G)} d(v) \geq \frac{1}{2} (n+1) \bigl(f(n+1) - f(n) + 2\bigr). $ On the other hand, we have $(n-1)(f(n+1)+1) - (n+1)f(n) = k$ with $k=0,-1$, which implies $f(n)=\frac{n-1}{n+1}(f(n+1)+1)-\frac{k}{n+1}$. Consequently, $|E(G)| > f(n+1) + 2$, contradicting the assumption. Therefore, there must exist a vertex $v_0$ in $G$ whose degree satisfies $d(v_0) \leq f(n+1) - f(n) + 1$. Consider the graph $G - v_0$, which has $|V(G-v_0)| = n$ vertices and $|E(G-v_0)| \geq f(n) + 1$ edges. Since $ex(n, H) = f(n)$, the graph $G-v_0$ contains a subgraph isomorphic to $H$, and hence so does $G$. Suppose that a graph $G$ has $n+1$ vertices and $f(n+1)+1$ edges. We show that the minimum degree of $G$ must be $f(n+1)-f(n)+1$. First, we can easily obtain $\delta(G)\geq f(n+1)-f(n)+1$; otherwise, the graph G with the vertex of minimum degree removed must contain a subgraph isomorphic to H. Secondly, we can obtain that $\delta(G)\leq f(n+1)-f(n)+1$; otherwise, this contradicts $|E(G)|=f(n+1)+1$. In conclusion, we have proven 2).
    
    For 3), suppose that a graph $G$ has $n+1$ vertices and $f(n+1)$ edges. We show that minimum degree of $G$ must be $f(n+1)-f(n)$. Otherwise, if every vertex $v$ in $G$ satisfies $d(v) \geq f(n+1) - f(n) + 1$, then the number of edges in $G$ satisfies $|E(G)| = \frac{1}{2} \sum_{v \in V(G)} d(v) \geq \frac{1}{2} (n+1) \bigl(f(n+1) - f(n)\bigr)$. We also have $(n-1)(f(n+1)+1) - (n+1)f(n) = k$ with $k\geq -1$, which implies $f(n)=\frac{n-1}{n+1}(f(n+1)+1)-\frac{k}{n+1}$. Consequently, $|E(G)| > f(n+1)$, contradicting the assumption. Therefore, there must exist a vertex $v_0$ in $G$ whose degree satisfies $d(v_0) \leq f(n+1) - f(n)$. On the other hand, if $d(v_0) \leq f(n+1) - f(n) -1$the graph G with the vertex of minimum degree removed must contain a subgraph isomorphic to H. Then we can obtain that $\delta(G)\leq f(n+1)-f(n)$. In conclusion, we have proven 3).
\end{IEEEproof}

Lemma~\ref{lemadd} determines the exact values of $ex(5, G_1)$, which is used to derive $ex(n, G_1)$.

\begin{restatable}{lemma}{GaddLemma}\label{lemadd}
$ex(5, G_1) = 7$, and the $5$-order extremal graph is $K_2\vee \overline{K_3}$ , $(K_2+K_1)\vee \overline{K_2}$ or the graph $G'_1$ generated by adding an edge to $K_4$. 
\end{restatable}

\begin{IEEEproof}
    It is easy to check that the graph $G$ obtained by deleting any two edges from $K_5$ is not $G_1$-free. Hence $ex(5,G_1)\leq7$. On the other hand, among the graphs obtained by deleting any three edges from $K_5$, only $K_2\vee \overline{K_3}$ , $(K_2+K_1)\vee \overline{K_2}$ and the graph $G'_1$ generated by adding an edge to $K_4$ are $G_1$-free, which implies $ex(5,G_1)\geq7$. This proves Lemma~\ref{lemadd}.
    \begin{figure}[!h]
    \centering
    \includegraphics[width=1in]{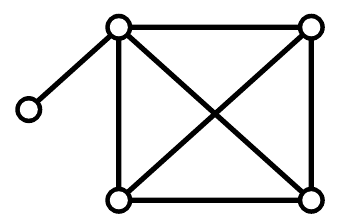}
    \caption{$G'_1$ is the 5-order extremal graph of $G_1$.}
    \label{figadd1}
    \end{figure}
\end{IEEEproof}

Theorem~\ref{th3} presents the exact Tur\'{a}n number for $G_1$ and the corresponding extremal graphs.

\begin{theorem}\label{th3}
    For all $n \geq 6$, $ex(n, G_1)=\lfloor\frac{n^2+n}{4}\rfloor.$
    Moreover, the $n$-order extremal graph $G$ depends on $n\pmod 4$:
\begin{enumerate}
    \item if $n = 4t$, then $G = tK_2 \vee \overline{K_{2t}}.$
    \item if $n = 4t + 1$, then $G = tK_2 \vee \overline{K_{2t+1}}$ or $G = (tK_2 + K_1) \vee \overline{K_{2t}};$
    \item if $n = 4t + 2$, then $G = (t+1)K_2 \vee \overline{K_{2t}}$ or $G = (tK_2 + K_1) \vee \overline{K_{2t+1}};$
    \item if $n = 4t + 3$, then $G = (t+1)K_2 \vee \overline{K_{2t+1}}.$
\end{enumerate}
\end{theorem}

\begin{IEEEproof}
We prove $ex(n,G_1)=\lfloor\frac{n^2+n}{4}\rfloor$ by induction and provide the precise characterization of extremal graphs.

For $n=6$, it is easy to prove $ex(6,G_1)=10$ and $(K_2+K_1)\vee \overline{K_3}$ , $2K_2\vee \overline{K_2}$ are $G_1$-free 6-order extremal graph.

Assuming the theorem holds for $n=k$ ($k\geq6$), we now consider the case $n=k+1$. Let $f(k)=\lfloor\frac{k^2+k}{4}\rfloor$ and $F(k)=(k-1)(f(k+1)+1)-(k+1)f(k)$. By direct calculation, we can obtain
$$
F(k)=
\begin{cases}
-1, & \text{if }k\equiv0\pmod4,\\[4pt]
\dfrac{k-1}{2}, & \text{if }k\equiv1\pmod4,\\[8pt]
k-1, & \text{if }k\equiv2\pmod4,\\[4pt]
\dfrac{k-3}{2}, & \text{if }k\equiv3\pmod4.
\end{cases}
$$

We consider two cases according to the value of $F(k)$.

\textbf{Case 1: } $k\equiv1,2,3\pmod4$.

We give the proof for $k\equiv1\pmod4$; the proofs for $k\equiv2,3\pmod4$ are similar to $k\equiv1\pmod4$.

For $k\equiv 1 \pmod4$, where $k=4t+1$, we have $F(k)>0$ and ${ex}(k,G_1)=f(k)$. By Lemma~\ref{lem1} 1), we obtain
${ex}(k+1,G_1)\leq f(k+1)$. On the other hand, $(t+1)K_2\vee \overline{K_{2t}}$ and $(tK_2+K_1)\vee \overline{K_{2t+1}}$ are $k+1$-vertex and $G_1$-free graphs with $f(k+1)$ edges. Hence ${ex}(k+1,G_1)\geq f(k+1)$. Therefore, we prove ${ex}(k+1,G_1)=f(k+1)$.

Suppose that $G$ is a $G_1$-free graph with $k+1$ vertices and $f(k+1)$ edges. Then there exists a vertex $v_0$ with minimum degree satisfying $d(v_0)=f(k+1)-f(k)$. By the induction hypothesis, $G-v_0=tK_2\vee \overline{K_{2t+1}}$ or $(tK_2+K_1)\vee \overline{K_{2t}}$. Under the condition of being $G_1$-free, the $(k+1)$-order extremal graph must be obtained from the $k$-order extremal graph $tK_2\vee \overline{K_{2t+1}}$ (or $(tK_2+K_1)\vee \overline{K_{2t}}$) by adding a vertex of degree $d_0=2t+1$. Consequently, $G=(t+1)K_2\vee \overline{K_{2t}}$ or $G=(tK_2+K_1)\vee \overline{K_{2t+1}}$.

\textbf{Case 2: } $k\equiv0\pmod4$.

For $k\equiv 4\pmod 4$, where $k=4t+4$, we have $F(k)=-1$ and ${ex}(k,G_1)=f(k)$. By Lemma~\ref{lem1} 2), we obtain ${ex}(k+1,G_1)\leq f(k+1)+1$. Consider a $G_1$-free graph $G$ with $k+1$ vertices and $f(k+1)+1$ edges. Then there exists a vertex $v_0$ with minimum degree satisfying $d(v_0)=f(k+1)-f(k)+1$. By the induction hypothesis, $G-v_0=(t+1)K_2\vee \overline{K_{2t+2}}$. In fact, the graph obtained from $(t+1)K_2\vee \overline{K_{2t+2}}$ by adding a vertex $v_0$ of degree $d(v_0)=2t+3$ cannot be $G_1$-free. This shows that $ex(k+1,G_1)\leq f(k+1)$. On the other hand, $(t+1)K_2\vee \overline{K_{2t+3}}$ and $((t+1)K_2+K_1)\vee \overline{K_{2t+2}}$ are $k+1$-vertex and $G_1$-free graphs with $f(k+1)$ edges. Hence ${ex}(k+1,G_1)\geq f(k+1)$. Therefore, we prove $ex(k+1,G_1)=f(k+1)$.

Suppose that $G$ is a $G_1$-free graph with $k+1$ vertices and $f(k+1)$ edges. Then there exists a vertex $v_0$ with minimum degree satisfying $d(v_0)=f(k+1)-f(k)$. By the induction hypothesis, $G-v_0=(t+1)K_2\vee \overline{K_{2t+2}}$. Under the condition of being $G_1$-free, the $(k+1)$-order extremal graph must be obtained from the $k$-order extremal graph $(t+1)K_2\vee \overline{K_{2t+2}}$ by adding a vertex of degree $d_0=2t+2$. Consequently, $G=(t+1)K_2\vee \overline{K_{2t+3}}$ or $G=((t+1)K_2+K_1)\vee \overline{K_{2t+2}}$.

By induction, the theorem is proved.
\end{IEEEproof}

Lemma~\ref{lem2} determines the exact values of $ex(6, G_2)$ and $ex(7,G_2)$, respectively, which are used to derive $ex(n, G_2)$.

\begin{restatable}{lemma}{GtwoLemma}\label{lem2}
The following statements are true:
\begin{enumerate}
    \item $ex(6, G_2) = 11$, and the $6$-order extremal graph is the graph $G_2'$ generated by adding an edge to $K_5$;
    \item $ex(7, G_2) = 13$, and the $7$-order extremal graph is $(K_2 + \overline{K}_{2} ) \vee \overline{K}_{3}$, $(K_1 + K_{2} ) \vee \overline{K}_{4}$, graph $G_2''$ generated by adding a path of length 3 to connect two non-adjacent vertices to $K_5$, or the graph $G_2'''$ which is composed of $K_3$ and $K_5$ with only one common vertex; 
    \item $ex(8, G_2) = 17$, and the $8$-order extremal graph is $(K_2 + \overline{K}_{2} ) \vee \overline{K}_{4}$. 
\end{enumerate}
\end{restatable}

The proof of Lemma~\ref{lem2} is relegated to Appendix A.

\begin{theorem}\label{th4}
 For all $n \geq 8$, $ex(n, G_2)=\lfloor\frac{n^2}{4}\rfloor+1.$ Moreover, the $n$-order extremal graph $G$ depends on $n\pmod 2$:
\begin{enumerate}
    \item if $n = 2t$, then $G = (K_2+\overline{K_{t-2}}) \vee \overline{K_t}$;
    \item if $n = 2t + 1$, then $G = (K_2+\overline{K_{t-1}}) \vee \overline{K_t}$ or $G = (K_2+\overline{K_{t-2}}) \vee \overline{K_{t+1}}$.
\end{enumerate}
\end{theorem}

\begin{IEEEproof}
We prove $ex(n, G_2) = \lfloor \frac{n^2}{4} \rfloor + 1$ by induction. 

For the base case $n = 8$, it follows from the Lemma~\ref{lem2} that $ex(8, G_2) = 17$.

Assuming the theorem holds for $n = k$ ($k \geq 8$), we now consider the case $n=k+1$. Let $f(k) = \lfloor \frac{k^2}{4} \rfloor + 1$ and $F(k) = (k - 1)(f(k + 1) + 1) - (k + 1)f(k)$. By direct calculation, we obtain
$$
F(k) = 
\begin{cases}
\displaystyle \frac{k}{2} - 3, & \text{if } k \text{ is even}, \\
\displaystyle k - 3, & \text{if } k \text{ is odd}.
\end{cases}
$$
For $k \geq 8$, we have $F(k) > 0$ and $ex(k, G_2) = f(k)$. By Lemma~\ref{lem1} 1), we obtain $ex(k + 1, G_2) \leq f(k + 1)$. On the other hand, $( K_2 + \overline{K}_{\lceil \frac{k + 1}{2} \rceil - 2} ) \vee \overline{K}_{\lfloor \frac{k + 1}{2} \rfloor}$ is a $G_2$-free graph with $k + 1$ vertices and $f(k + 1)$ edges. Hence, $ex(k + 1, G_2) \geq f(k + 1)$. Therefore, we prove $ex(k + 1, G_2) = f(k + 1)$. Similar to Theorem~\ref{th2}, all cases of $k+1$-order extremal graphs for $G_2$ can be readily determined. By induction, the theorem is proved.
\end{IEEEproof}

The following lemmas serve as the foundation for deriving the Tur\'{a}n number of $G_3$ and characterizing its extremal graphs. We first state Lemma~\ref{lem3} without proof.

\begin{lemma}\label{lem3}
    Let $G$ be a simple 3-regular graph of order $6$. Then $G=K_{3,3}$ or $G=H_{3}$. 
\end{lemma}

\begin{figure}[!h]
\centering
\includegraphics[width=2in]{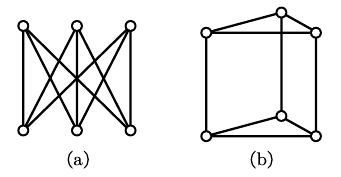}
\caption{The complete bipartite graph $K_{3,3}$ and the triangular prism graph $H_3$.}
\end{figure}

\begin{restatable}{lemma}{GfourLemma}\label{lem4}
The following statements are true:
    \begin{enumerate}
        \item $ex(6, G_3) = 12$, and the $6$-order extremal graph is $G = K_3 \vee \overline{K_3}$;
        \item $ex(7, G_3) = 15$, and the $7$-order extremal graph is $G = (K_3 + K_1) \vee \overline{K_3}$, $G = K_3 \vee \overline{K_4}$, or $G = S_4 \vee \overline{K_3}$;
        \item $ex(8, G_3) = 19$, and the $7$-order extremal graph is $G = (K_3 + K_1) \vee \overline{K_4}$ or $G = S_4 \vee \overline{K_4}$.
    \end{enumerate}
\end{restatable}

The proof of Lemma~\ref{lem4} is relegated to Appendix A.

\begin{theorem}\label{th5}
For all $n \geq 9$, $ex(n, G_3) = \lfloor \frac{n^2}{4} \rfloor + \lceil \frac{n}{2} \rceil - 1.$
Moreover, the $n$-order extremal graph $G$ depends on $n\pmod 2$:
\begin{enumerate}
    \item if $n = 2t$, then $G = S_t \vee \overline{K_t}$;
    \item if $n = 2t + 1$, then $G = S_{t+1} \vee \overline{K_t}$.
\end{enumerate} 
\end{theorem}
\begin{IEEEproof}
We prove $ex(n, G_3) = \lfloor \frac{n^2}{4} \rfloor + \lceil \frac{n}{2} \rceil - 1$ by induction and provide the precise characterization of extremal graphs.

For the base case $n=9$, it is easy to prove $ex(9,G_3)=24$ and $S_5\vee \overline{K_4}$is the only $G_3$-free 9-order extremal graph. 

Assuming the theorem holds for $n = k$ ($k \geq 9$), we now consider the case $n=k+1$. Let $f(k) = \lfloor \frac{k^2}{4} \rfloor + \lceil \frac{k}{2} \rceil - 1$ and $F(k) = (k - 1)(f(k + 1) + 1) - (k + 1)f(k)$. By direct calculation, we can obtain
$$
F(k) = 
\begin{cases}
\displaystyle \frac{k}{2}, & \text{if } k \text{ is even}, \\
\displaystyle 0, & \text{if } k \text{ is odd}.
\end{cases}
$$

We consider two cases according to the value of $F(k)$.

\textbf{Case 1:} $k$ is even.  

For $k$ is even, where $k = 2t$, we have $F(k) > 0$ and $ex(k, G_3) = f(k)$. By Lemma~\ref{lem1} 1), we obtain $ex(k + 1, G_3) \leq f(k + 1)$. On the other hand, $S_{t+1} \vee \overline{K_t}$ is a $G_3$-free graph with $k + 1$ vertices and $f(k + 1)$ edges, so $ex(k + 1, G_3) \geq f(k + 1)$. Thus we have proved $ex(k + 1, G_3) = f(k + 1)$.  

Suppose that $G$ is a $G_3$-free graph with $k + 1$ vertices and $f(k + 1)$ edges. Then there exists a vertex $v_0$ with minimum degree satisfying $d(v_0) = f(k + 1) - f(k)$. By the induction hypothesis, $G - v_0 = S_t \vee \overline{K_t}$. Under the condition of being $G_3$-free, the $k+1$-order extremal graph must be obtained from the $k$-order extremal graph $S_t \vee \overline{K_t}$ by adding a vertex of degree $d_0 = t + 1$. Consequently, $G = S_{t+1} \vee \overline{K_t}$.

\textbf{Case 2:} $k$ is odd. 

For $k$ is odd, where $k = 2t + 1$, we have $F(k) = 0$ and $ex(k, G_3) = f(k)$. By Lemma~\ref{lem1} 2), we obtain $ex(k + 1, G_3) \leq f(k + 1) + 1$. Consider a $G_3$-free graph $G$ with $k + 1$ vertices and $f(k + 1) + 1$ edges. Then there exists a vertex $v_0$ with minimum degree satisfying $d(v_0) = f(k + 1) - f(k) + 1$. By the induction hypothesis, $G - v_0 = S_{t+1} \vee \overline{K_t}$. In fact, the graph obtained from $S_{t+1} \vee \overline{K_t}$ by adding a vertex $v_0$ of degree $d(v_0) = t + 2$ cannot be $G_3$-free. This shows that $ex(k + 1, G_3) \leq f(k + 1)$. On the other hand, $S_{t+1} \vee \overline{K_{t+1}}$ is a $G_3$-free graph with $k + 1$ vertices and $f(k + 1)$ edges, so $ex(k + 1, G_3) \geq f(k + 1)$. Thus we have proved $ex(k + 1, G_3) = f(k + 1)$.  

Suppose that $G$ is a $G_3$-free graph with $k + 1$ vertices and $f(k + 1)$ edges. Then there exists a vertex $v_0$ with minimum degree satisfying $d(v_0) = f(k + 1) - f(k)$. By the induction hypothesis, $G - v_0 = S_{t+1} \vee \overline{K_t}$. Under the condition of being $G_3$-free, the $k+1$-order extremal graph of order $k + 1$ for $G_3$ must be obtained from the extremal graph of order $k$, $S_{t+1} \vee \overline{K_t}$, by adding a vertex of degree $d_0 = t + 1$. Therefore, $G = S_{t+1} \vee \overline{K_{t+1}}$.

In summary, the theorem is proved by induction.
\end{IEEEproof}

\section{Theoretical bounds for ETSs}
In this section, we study ETSs in Tanner graph by Tur\'{a}n numbers and establish a connection between specific structures in VN graphs and ETSs in Tanner graph. We also provide theoretical bounds for elementary trapping sets.

When the girth is 8 and if any two 8-cycles in Tanner graph do not share a common variable node, then the VN graph of any (a, b)-ETSs is $\{\theta(2,2,2), \theta(1, 3, 3), D(4,4;0)\}$-free. This leads to the following theorem:

\begin{figure}[!t]
\centering
\includegraphics[width=2.9in]{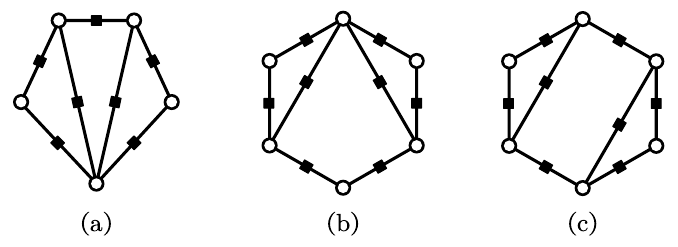}
\caption{Figures (a), (b) and (c) depict three substructures $T_1$, $T_2$ and $T_3$ of the Tanner graph and the corresponding VN graphs are $G_1$, $G_2$ and $G_3$,  respectively.}
\label{fig4}
\end{figure}

\begin{theorem}\label{th6}
    For an $(a,b)$-ETS in a Tanner graph with girth 8 and variable-regular degree $d_L(v) = \gamma$, if any two 8-cycles share no common variable node, then $b\geq a\gamma-\frac{a(\sqrt{24a-23}-1)}{4}$.
\end{theorem}

\begin{IEEEproof}
    For the VN graph $G_{VN}=(V_{VN},E_{VN})$ of an $(a,b)$-ETS with variable-regular degree of $\gamma$, we have $|V_{VN}|=a$, $E_{VN}=\frac{a\gamma-b}{2}$. Since there is no $\theta(2,2,2), \theta(1, 3, 3), D(4,4;0)$ in $G_{VN}$, $|E(VN)|$ is bounded above by $ex(a,\{C_3, \theta(2,2,2), \theta(1, 3, 3), D(4,4;0)\})$. Based on Theorem~\ref{th2}, we can obtain $E_{VN}=\frac{a\gamma-b}{2}\leq\frac{a(\sqrt{24a-23}-1)}{8}$, which further implies $b\geq a\gamma-\frac{a(\sqrt{24a-23}-1)}{4}$. 
\end{IEEEproof}

For a given value of $\gamma$, applying this bound allows us to theoretically demonstrate the non-existence of certain small ETSs. Although the minimum value of a determined by the bound may sometimes be less than the exact value, using this bound can significantly reduce the complexity of enumeration. Taking into account the parity conditions that $a$ and $b$ must satisfy, Corollary~\ref{co7} and Corollary~\ref{co8} present the minimum values of a derived from the bound $b\geq a\gamma-\frac{a(\sqrt{24a-23}-1)}{4}$ for $\gamma = 3$ and $\gamma = 4$, respectively. These values correspond to the feasible sizes of an $(a, b)$-ETS under our conditions.

\begin{corollary}\label{co7}
    For LDPC codes with girth 8 and $\gamma = 3$, within a Tanner graph where any two 8-cycles share no common variable node, if $b < a$, then the smallest sizes of an $(a, b)$-ETS are $a = 10, 11, 10, 7$ for $b = 0, 1, 2, 3$, respectively.
\end{corollary}
\begin{corollary}\label{co8}
    For LDPC codes with girth 8 and $\gamma = 4$, within a Tanner graph where any two 8-cycles share no common variable node, if $b < a$, then the smallest sizes of an $(a, b)$-ETS are at least $a = 13, 13, 12, 10$ for $b = 0, 2, 4, 6$, respectively.
\end{corollary}

\begin{figure}[!t]
\centering
\includegraphics[width=3.1in]{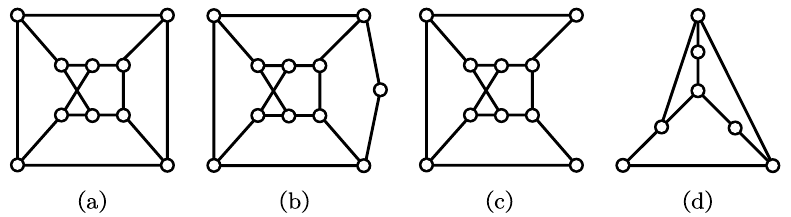}
\caption{ Figures (a), (b), (c), and (d) depict the VN graphs of (10,0), (11,1), (10,2) and (7,3)-ETSs, respectively, in the Tanner graph with $g=8$ and $\gamma=3$.}
\label{fig5}
\end{figure}

It is worth noting that in Corollary~\ref{co8}, we particularly emphasize 'at least' because the existence of the corresponding $(a,b)$-ETS has not been verified. As the values of $a$, $b$ and $\gamma$ increase, verifying the existence of the $(a,b)$-ETS becomes increasingly difficult.

The VN graphs corresponding to the substructures $T_1$, $T_2$, and $T_3$ in Tanner graph are $G_1$, $G_2$, and $G_3$, respectively. When the girth is 6 and if Tanner graph is $T_i $ free$(i=1,2,3)$, the VN graph is $G_i$ free. This leads to the following theorem.

\begin{theorem}\label{th9}
    For an $(a,b)$-ETS in a Tanner graph with girth 6 and variable-regular degree $d_L(v) = \gamma$, then the following statements hold:
    \begin{enumerate}
    \item if the ETS is $T_1$-free, then for $a \geq 5$, we have $b \geq a \gamma -2\lfloor\frac{a^2+a}{4}\rfloor$;
    \item if the ETS is $T_2$-free, then for $a \geq 7$, we have $b \geq a \gamma -2\lfloor\frac{a^2}{4}\rfloor-2$;
    \item if the ETS is $T_3$-free, then for $a \geq 7$, we have $b \geq a \gamma -2\lfloor\frac{a^2}{4}\rfloor-2\lceil \frac{a}{2}\rceil+2$.
\end{enumerate}
\end{theorem}

When the Tanner graph is free of $T_1$, $T_2$ or $T_3$, we determine the minimum value of $a$ corresponding to different values of $b$. These results are presented in Corollary~\ref{co10}.
\begin{corollary}\label{co10}
For LDPC codes with girth 6 and $\gamma = 3, 4, 5$, within a Tanner graph which is $T_i(i=1,2,3)$ free, if $b < a$, then the smallest sizes of $(a, b)$-ETSs are shown in the Table \ref{tab1}.    
\end{corollary}
\begin{table}[h]
    \centering
    
    \caption{The minimum size of an $(a, b)$-ETS, where $b < a$, within an LDPC code of girth 6 and variable-regular degree $\gamma$, ensures that the Tanner graph is $T_i$-free. }
    
    \begin{tabular}{|c|c|c|c|c|}
        \hline
        $\gamma$ & $b$ & $T_1$-free & $T_2$-free & $T_3$-free \\ \hline
        \multirow{4}*{3} & 0 & 4 & 4 & 4 \\ \cline{2-5}
        & 1 & 5 & 5 & 5 \\ \cline{2-5}
        & 2 & 4 & 4 & 4 \\ \cline{2-5}
        & 3 & 3 & 3 & 3 \\ \hline
        \multirow{3}*{4} & 0 & 7 & 5 & 5 \\ \cline{2-5}
        & 2 & 7 & 5 & 5 \\ \cline{2-5}
        & 4 & 6 & 4 & 4 \\ \hline
        \multirow{8}*{5} & 0 & 10 & 10 & 10 \\ \cline{2-5}
        & 1 & 9 & 11 & 11 \\ \cline{2-5}
        & 2 & 10 & 10 & 10 \\ \cline{2-5}
        & 3 & 9 & 9 & 9 \\ \cline{2-5}
        & 4 & 8 & 10 & 10 \\ \cline{2-5}
        & 5 & 9 & 9 & 9 \\ \cline{2-5}
        & 6 & 8 & 8 & 8 \\ \hline
    \end{tabular}
    
    \label{tab1}
\end{table}

We point out that for LDPC codes with girth $g = 6$ and variable regular degree $d_L(v) = 3$, if the Tanner graph is free of the structures $T_i$ ($i = 1, 2, 3$), then the minimum sizes of $(a,b)$-ETS corresponding to given $b = 0, 1, 2, 3$ are $a = 4, 5, 4, 3$, respectively. The minimum sizes do not satisfy the inequalities in Theorem~\ref{th9}. The reason is that theorem restricts the range of the positive integer $a$, and when $a$ is relatively small the bounds in the theorem are not applicable. The problem does not occur when $d_L(v) \geq 5$.

The theoretical bounds derived from the Tur\'{a}n numbers of special graphs reflect the connection between these structures and ETS. This approach reduces the computational complexity and provides theoretical guidance for improving the error floor phenomenon of LDPC codes.

\section{Spectral radius analysis and simulation}

In this section, we further investigate the effects of these structural constraints from two perspectives. First, we analyze the spectral radii of the system matrices associated with ETSs under the proposed structural restrictions. Second, we construct QC-LDPC codes that satisfy these restrictions and evaluate their error-floor performance through simulations.

\subsection{Spectral Radius of ETSs}

According to the linear state-space model for ETSs, the spectral radius of the associated system matrix is a key parameter that affects decoding behavior. In particular, a larger spectral radius generally indicates that erroneous messages within an ETS decay more slowly during iterative decoding, thereby making the ETS more harmful in the error floor region. Motivated by this observation, we examine whether the structural constraints proposed in Section IV also reduce the spectral radii of the remaining ETSs.

We first consider Tanner graphs in which no two 8-cycles share a common variable node. In this case, the corresponding VN graph of an ETS contains no pair of 4-cycles sharing at least one common vertex. To quantify the effect of this restriction, we partition all $(a,b)$-ETSs into two classes, denoted by $\mathrm{int}_{C_4}(a,b)$ and $\mathrm{ind}_{C_4}(a,b)$. The set $\mathrm{int}_{C_4}(a,b)$ consists of ETSs whose VN graphs contain one of the subgraphs $\theta(1,3,3)$, $\theta(2,2,2)$, or $D(4,4;0)$, whereas $\mathrm{ind}_{C_4}(a,b)$ consists of ETSs whose VN graphs are free of all these subgraphs. For each ETS in these two classes, we compute the spectral radius of its system matrix and then report the corresponding mean and median values.

Using the graph tool `nauty' \cite{nauty}, we enumerated all non-isomorphic VN graphs satisfying the required conditions. Table \ref{tab2} summarizes the results for all $(a, b)$-ETSs with $\gamma = 3$, $a \leq 14$, and $b \leq 6$. For all listed parameter pairs, both the mean and the median spectral radii of the ETSs in $\text{ind}_{C_4}(a,b)$ are smaller than those of the ETSs in $\text{int}_{C_4}(a,b)$. This observation suggests that ETSs whose VN graphs avoid the above substructures tend to be less harmful from the perspective of spectral radius. Therefore, imposing the condition that 8-cycles in the Tanner graph do not share at least one variable node can effectively eliminate ETSs with relatively large spectral radii and is thus beneficial to the error floor performance of LDPC codes.

Next, we consider the short cycles with chords $T_i$, where $i \in \{1, 2, 3\}$. For each $T_i$, we divide all $(a, b)$-ETSs into two sets: $\mathcal{F}_{G_i}(a,b)$, consisting of ETSs whose VN graphs are $G_i$-free, and $\mathcal{U}_{G_i}(a,b)$, consisting of ETSs whose VN graphs contain $G_i$ as a subgraph. For each ETS in these two sets, we compute the spectral radius of the associated system matrix and the corresponding mean and median values for each class.

\begin{table*}[!t]
\centering
\caption{The Number of $(a, b)$-ETSs in $\text{ind}_{C_4}(a, b)$ and $\text{int}_{C_4}(a, b)$,  
and the Corresponding Mean and Median Values of the Spectral Radius for $a, b$, and $\gamma = 3$}
\label{tab2}
\begin{tabular}{cccc|cccc|cccc}
\toprule
Sets of ETSs & Num. & $\rho_\text{mean}$ & $\rho_\text{median}$ & Sets of ETSs & Num. & $\rho_\text{mean}$ & $\rho_\text{median}$ & Sets of ETSs & Num. & $\rho_\text{mean}$ & $\rho_\text{median}$ \\
\toprule
$\text{ind}_{C_4}$(10,2) & 8  & 1.82994 & 1.82982 & $\text{ind}_{C_4}$(11,5) & 127 & 1.58806 & 1.58560 & $\text{ind}_{C_4}$(13,3) & 603  & 1.81013 & 1.80602 \\
$\text{int}_{C_4}$(10,2) & 20 & 1.84185 & 1.83727 & $\text{int}_{C_4}$(11,5) & 54  & 1.61762 & 1.61279 & $\text{int}_{C_4}$(13,3) & 494 & 1.82677 & 1.82408 \\
\midrule
$\text{ind}_{C_4}$(10,4) & 28 & 1.63866 & 1.63725 & $\text{ind}_{C_4}$(12,2) & 83  & 1.86560 & 1.86271 & $\text{ind}_{C_4}$(13,5) & 1369 & 1.66501 & 1.66170 \\
$\text{int}_{C_4}$(10,4) & 35 & 1.66093 & 1.65234 & $\text{int}_{C_4}$(12,2) & 107 & 1.87291 & 1.86741 & $\text{int}_{C_4}$(13,5) & 597  & 1.69008 & 1.68172 \\
\midrule
$\text{ind}_{C_4}$(10,6) & 42 & 1.43289 & 1.43199 & $\text{ind}_{C_4}$(12,4) & 322 & 1.71159 & 1.71200 & $\text{ind}_{C_4}$(14,2) & 879  & 1.88828 & 1.88514 \\
$\text{int}_{C_4}$(10,6) & 7 & 1.46695 & 1.46185 & $\text{int}_{C_4}$(12,4) & 217 & 1.73294 & 1.72456 & $\text{int}_{C_4}$(14,2) & 777  & 1.89284 & 1.88830 \\
\midrule
$\text{ind}_{C_4}$(11,1) & 8  & 1.92901 & 1.92899 & $\text{ind}_{C_4}$(12,6) & 402 & 1.54529 & 1.54220 & $\text{ind}_{C_4}$(14,4) & 3697 & 1.76194 & 1.76038 \\
$\text{int}_{C_4}$(11,1) & 15 & 1.92988 & 1.92934 & $\text{int}_{C_4}$(12,6) & 110 & 1.57941 & 1.57147 & $\text{int}_{C_4}$(14,4) & 2087 & 1.77728 & 1.77106 \\
\midrule
$\text{ind}_{C_4}$(11,3) & 56 & 1.76573 & 1.76264 & $\text{ind}_{C_4}$(13,1) & 75  & 1.94118 & 1.94103 & $\text{ind}_{C_4}$(14,6) & 5084 & 1.62441 & 1.62030  \\
$\text{int}_{C_4}$(11,3) & 66 & 1.78342 & 1.77965 & $\text{int}_{C_4}$(13,1) & 87  & 1.94213 & 1.94153 & $\text{int}_{C_4}$(14,6) & 1521 & 1.65284 & 1.64367 \\
\bottomrule
\end{tabular}
\end{table*}

\begin{table*}[!t]
\centering
\caption{The Number of $(a, b)$-ETSs in $\mathcal{F}_{G_i}(a, b)$ and $\mathcal{U}_{G_i}(a, b)$ with $i\in \{ 1,2,3\}$,  
and the Corresponding Mean and Median Values of the Spectral Radius for $a, b$, and $\gamma = 4$}
\label{tab3}
\begin{tabular}{cccc|cccc|cccc}
\toprule
Set of ETSs & Num. & $\rho_{\text{mean}}$ & $\rho_{\text{median}}$ & Set of ETSs & Num. & $\rho_{\text{mean}}$ & $\rho_{\text{median}}$ & Set of ETSs & Num. & $\rho_{\text{mean}}$ & $\rho_{\text{median}}$ \\
\toprule
$\mathcal{F}_{G_1}$(8,2)  & 13   & 2.79279 & 2.78237 & $\mathcal{F}_{G_2}$(8,2)  & 9    & 2.78968 & 2.77971 & $\mathcal{F}_{G_3}$(8,2)  & 4    & 2.79741 & 2.78264\\
$\mathcal{U}_{G_1}$(8,2)  & 22   & 2.80125 & 2.79270 & $\mathcal{U}_{G_2}$(8,2)  & 26   & 2.80103 & 2.79231 & $\mathcal{U}_{G_3}$(8,2)  & 31   & 2.79820 & 2.79191\\
\midrule
$\mathcal{F}_{G_1}$(8,4)  & 55   & 2.56767 & 2.55509 & $\mathcal{F}_{G_2}$(8,4)  & 40   & 2.56995 & 2.55627 & $\mathcal{F}_{G_3}$(8,4)  & 28   & 2.59767 & 2.58721\\
$\mathcal{U}_{G_1}$(8,4)  & 69   & 2.60658 & 2.59440 & $\mathcal{U}_{G_2}$(8,4)  & 84   & 2.59855 & 2.59218 & $\mathcal{U}_{G_3}$(8,4)  & 96   & 2.58689 & 2.58730\\
\midrule
$\mathcal{F}_{G_1}$(9,2)  & 68   & 2.81956 & 2.80842 & $\mathcal{F}_{G_2}$(9,2)  & 47   & 2.82123 & 2.81139 & $\mathcal{F}_{G_3}$(9,2)  & 22   & 2.82206 & 2.81236\\
$\mathcal{U}_{G_1}$(9,2)  & 86   & 2.82500 & 2.81755 & $\mathcal{U}_{G_2}$(9,2)  & 107  & 2.82319 & 2.81058 & $\mathcal{U}_{G_3}$(9,2)  & 132  & 2.82269 & 2.81056\\
\midrule
$\mathcal{F}_{G_1}$(9,4)  & 342  & 2.62755 & 2.62685 & $\mathcal{F}_{G_2}$(9,4)  & 248  & 2.62909 & 2.62878 & $\mathcal{F}_{G_3}$(9,4)  & 164  & 2.63555 & 2.62467\\
$\mathcal{U}_{G_1}$(9,4)  & 321  & 2.65354 & 2.64828 & $\mathcal{U}_{G_2}$(9,4)  & 415  & 2.64673 & 2.64342 & $\mathcal{U}_{G_3}$(9,4)  & 499  & 2.64164 & 2.64134\\
\midrule
$\mathcal{F}_{G_1}$(10,2) & 438  & 2.83814 & 2.82928 & $\mathcal{F}_{G_2}$(10,2) & 257  & 2.83875 & 2.83034 & $\mathcal{F}_{G_3}$(10,2) & 175  & 2.83883 & 2.83068\\
$\mathcal{U}_{G_1}$(10,2) & 407  & 2.84299 & 2.83388 & $\mathcal{U}_{G_2}$(10,2) & 588  & 2.84123 & 2.83112 & $\mathcal{U}_{G_3}$(10,2) & 670  & 2.84090 & 2.83071\\
\midrule
$\mathcal{F}_{G_1}$(10,4) & 2441 & 2.67138 & 2.67301 & $\mathcal{F}_{G_2}$(10,4) & 1653 & 2.67219 & 2.67179 & $\mathcal{F}_{G_3}$(10,4) & 1225 & 2.67173 & 2.66428\\
$\mathcal{U}_{G_1}$(10,4) & 1732 & 2.69026 & 2.68738 & $\mathcal{U}_{G_2}$(10,4) & 2520 & 2.68383 & 2.68393 & $\mathcal{U}_{G_3}$(10,4) & 2948 & 2.68232 & 2.68323\\
\bottomrule
\end{tabular}
\end{table*}

The results for all $(a, b)$-ETSs with $\gamma = 4$, $a \leq 10$, and $b \leq 4$ are presented in Table~\ref{tab3}. For each $T_i$ and each listed parameter pair $(a, b)$, the mean and median spectral radii of the ETSs in $\mathcal{F}_{G_i}(a,b)$ are consistently smaller than those in $\mathcal{U}_{G_i}(a,b)$. This finding indicates that ETSs containing $G_i$ are generally more detrimental than those that avoid $G_i$. Consequently, excluding the substructures $T_i$ from the Tanner graph can help reduce the occurrence of ETSs with relatively large spectral radii, thereby improving the error floor performance.

\subsection{Simulation Results}

We evaluate whether the proposed criteria can be used to improve finite-length QC-LDPC codes. Since all the substructures are formed by intersecting short cycles, we use the QC-PEG-CYCLE algorithm \cite{Xiong2} to construct QC-LDPC codes. Under the same girth, base matrix, and lifting degree, we aim to construct QC‑LDPC codes with fewer target substructures to verify the theoretical results. For compactness, let $F_1,F_2,$ and $F_3$ denote the substructures of Tanner graph, whose VN graphs are $\theta(2,2,2)$, $\theta(1,3,3)$, and $D(4,4;0)$, respectively. 
Code $E_1$ is a (4, 7)-regular QC-LDPC code with girth 8 and lifting degree 77 which is constructed primarily to suppress $F_1,F_2,$ and $F_3$.  
Code $E_2$ is a (4, 8)-regular QC-LDPC code with girth 6 and lifting degree 35 which is constructed primarily to decrease the number of $T_1,T_2,$ and $T_3$. 
The exponent matrices of $E_1$ and $E_2$ are $H_1$ and $H_2$. The resulting finite-lift codes are therefore described as structure-suppressed, rather than strictly structure-free.

$$
H_1=\begin{bmatrix}
40 & 64 & 70 & 20 & 47 & 11 & 75 \\
68 & 0 & 8 & 4 & 39 & 51 & 2 \\
65 & 36 & 12 & 50 & 40 & 47 & 13 \\
54 & 29 & 49 & 59 & 2 & 18 & 14
\end{bmatrix}
$$
$$
H_2 =
\begin{bmatrix}
4 & 8 & 6 & 19 & 32 & 22 & 10 & 3 \\
33 & 10 & 10 & 26 & 19 & 19 & 1 & 6 \\
25 & 21 & 23 & 31 & 22 & 13 & 21 & 31 \\
23 & 16 & 1 & 20 & 14 & 29 & 32 & 8
\end{bmatrix}
$$

For $E_1$, the comparison set consists of codes constructed by the PEG algorithm \cite{PEG}, the GCD-based method \cite{GCD}, and the reversal of dual sequences method \cite{sfxl}. 
For $E_2$, the comparison set contains one code constructed by the PEG algorithm \cite{PEG} and two codes constructed using the ``Simple Form'' method in \cite{simpleform}. Table \ref{tab:code-structures1} and \ref{tab:code-structures2} reports the numbers of cycles and target substructures. 

\begin{table}[!t]
\caption{Numbers of $C_6$, $C_8$, $F_1$, $F_2$, and $F_3$ contained in $E_1$ and its comparison set.}
\label{tab:code-structures1}
\centering
\footnotesize
\setlength{\tabcolsep}{5.0pt}
\begin{tabular}{c|cc|ccc}
\hline
Code & $N_6$ & $N_8$ & $N_{F_1}$ & $N_{F_2}$ & $N_{F_3}$  \\ \hline
$E_1$ & 0 & 12320 & \textbf{1386} & \textbf{85316} & \textbf{376761} \\
RDS   & 0 & 25718 & 8932 & 374528 & 1759142 \\
PEG  & 924 & 16247 & 3465 & 149226 & 625779 \\
GCD   & 0 & 24486 & 6622 & 415415 & 1482096\\ \hline
\end{tabular}
\end{table}

\begin{table}[!t]
\caption{Numbers of $C_6$, $C_8$, $T_1$, $T_2$, and $T_3$ contained in $E_2$ and its comparison set.}
\label{tab:code-structures2}
\centering
\footnotesize
\setlength{\tabcolsep}{5.0pt}
\renewcommand{\arraystretch}{1}
\begin{tabular}{c|cc|ccc}
\hline
Code & $N_6$ & $N_8$ &  $N_{T_1}$ & $N_{T_2}$ & $N_{T_3}$ \\ \hline
$E_2$ & 525 & 29470 & \textbf{105} & \textbf{1575} & \textbf{1575} \\
PEG  & 1260 & 24395 & 420 & 6265 & 5775 \\
SF1,2   & 1750 & 26880 & 1680 & 19285 & 18725 \\
\hline
\end{tabular}
\end{table}
\begin{table}[!t]
\caption{Common simulation settings.}
\label{tab:simulation-settings}
\centering
\footnotesize
\setlength{\tabcolsep}{3pt}
\renewcommand{\arraystretch}{0.98}
\begin{tabular}{@{}cc@{}}
\hline
Parameter & Setting \\
\hline
Signal & All-zero word; BPSK\\
Channel & AWGN \\
Decoder & Layered normalized min-sum\\
Actual rates & $R_1=143/280$ and $R_2=234/539$ \\
Noise variance & $\sigma^2=[2R\,10^{(E_b/N_0)_{\rm dB}/10}]^{-1}$ \\
Per-SNR stop & $500$ errors or $10^8$ frames \\
\hline
\end{tabular}
\end{table}

The structural differences are substantial.  Relative to the best code in its comparison set, $E_1$ reduces $N_{F_1},N_{F_2},$ and $N_{F_3}$ by $60.0\%$, $42.8\%$, and $39.8\%$, respectively.  Likewise, $E_2$ reduces $N_{T_1},N_{T_2},$ and $N_{T_3}$ by $75.0\%$, $74.9\%$, and $72.7\%$, respectively. The common simulation settings are summarized in Table~\ref{tab:simulation-settings}.  The FER curves of the two code groups are shown in Figs.~\ref{FZ1} and~\ref{FZ2}.

\begin{figure}[!h]
\centering

\includegraphics[width=3in]{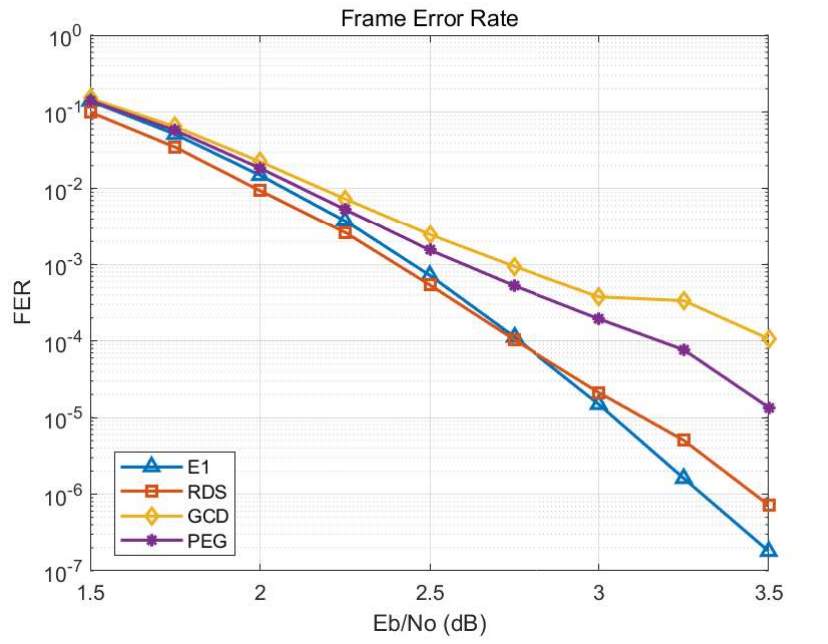}
\caption{FER performance of $E_1$ and its counterparts.}

\label{FZ1}
\end{figure}
\begin{figure}[!h]
\centering

\includegraphics[width=3in]{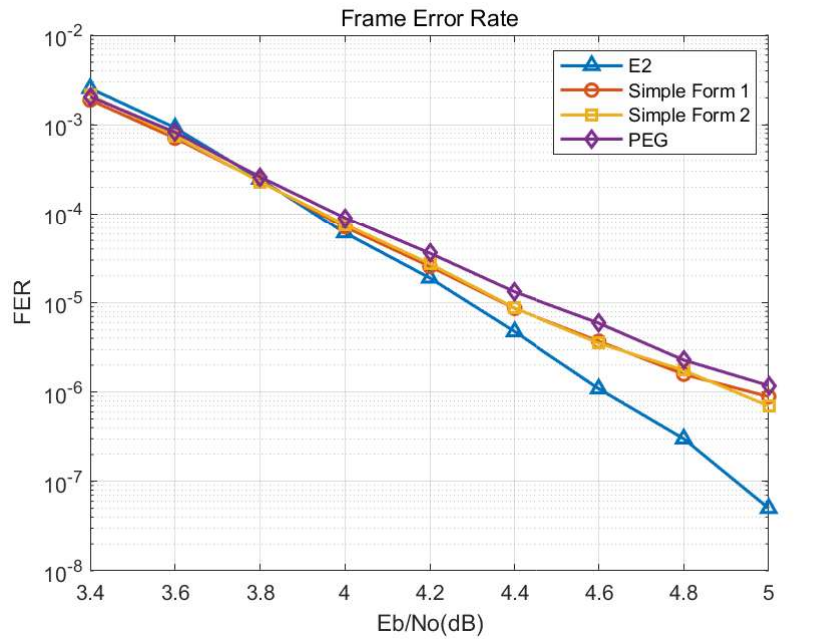}
\caption{FER performance of $E_2$ and its counterparts.}
\label{FZ2}
\end{figure}


Thus, the simulation results demonstrate the practical benefit of suppressing the targeted substructures, while the FER curves test whether it yields error-floor improvement.

\section{Conclusion}

In this paper, we derive the Tur\'{a}n numbers for several special graphs, including theta graphs, dumbbell graphs, and short cycles with chords, and translate these extremal results into structural constraints on ETSs when the corresponding subgraphs are forbidden in the VN graph. For a variable-regular LDPC code with girth \(g=8\) and variable degree \(d_L(v)=\gamma\), if configurations in which any two 8-cycles share a common variable node are excluded, then every \((a,b)\)-ETS must satisfy the bound \(b \geq a\gamma - \frac{a(\sqrt{24a-23}-1)}{4}\), which improves existing results on the minimum size of ETSs. In addition, avoiding certain types of short-cycle structures with chords, namely \(T_1\), \(T_2\), and \(T_3\), in the Tanner graph can effectively eliminate some classes of ETSs. These results provide useful theoretical support for the construction and performance analysis of LDPC codes. Furthermore, the spectral-radius analysis and simulation results demonstrate that the QC-LDPC codes constructed in this paper achieve a lower error floor, indicating improved performance.

\appendices

\section{}
\subsection{Proof of Lemma~\ref{lem2}}

\GtwoLemma*

\begin{IEEEproof}
    For 1), where $n = 6$, it is easy to check that the graph $G$ obtained by deleting any $3$ edges from $K_6$ is not $G_2$-free. Hence $ex(6, G_2) \leq 11$. On the other hand, if we remove any $4$ edges from $K_6$, only $G'_2$ is $G_2$-free, so $ex(6, G_2) \geq 11$. This proves that $ex(6, G_2) = 11$, and $|E(G)| = 11$ if and only if $G = G'_2$.

    \begin{figure}[!h]
    \centering
    \includegraphics[width=1in]{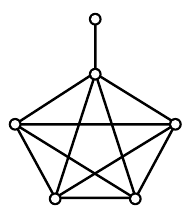}
    \caption{$G'_2$ is the 6-order extremal graph of $G_2$.}
    \label{fig6}
    \end{figure}

    For 2), suppose that $G$ is a $G_2$-free graph with $7$ vertices and 14 edges.
    
    First we show that the minimum degree of $G$, denoted by $d_1$, must be $4$. If $d_1 \geq 5$, then $|E(G)| \geq \frac{5 \times 7}{2} > 17$, which is a contradiction. If $d_1 \leq 3$, consider the graph $G - v_1$, which satisfies $|E(G - v_1)| \geq 14 - 3 = 11$. Since $G - v_1$ is also $G_2$-free, we have $G - v_1 = G'$. Note that the minimum degree of $G$ is $d_1 = 3$ while the minimum degree of $G - v_1$ is $1$, which is a contradiction. Thus $d_1 = 4$.

    Next we prove that such a graph $G$ must contain $G_2$ as a subgraph.  Because of $d_1=4$, the degree sequence of $G$ must be $\pi_G = (4,4,4,4,4,4,4)$. Let $v_1$ and $v_2$ be two non-adjacent vertices. We distinguish two cases.
    
    \textbf{Case 1:} $N_G(v_1) = N_G(v_2)$.
    
    Suppose that $N_G(v_1) = N_G(v_2) = \{v_3, v_4, v_5, v_6\}$. Since $d_7 = 4$, we have $N_G(v_7) = \{v_3, v_4, v_5, v_6\}$. To satisfy the degree sequence, there must be two non-adjacent edges among $\{v_3, v_4, v_5, v_6\}$, say $v_3v_4$ and $v_5v_6$. Consequently, the graph $G$ must be isomorphic to the graph shown in Fig.~\ref{fig7}, which contains $G_2$ as a subgraph.
    \begin{figure}[!h]
    \centering
    \includegraphics[width=1.5in]{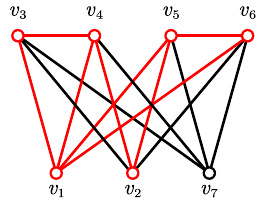}
    \caption{The graph $G$ satisfying $N_G(v_1) = N_G(v_2)$.}
    \label{fig7}
    \end{figure}

    \textbf{Case 2:} $N_G(v_1) \neq N_G(v_2)$.

    It is easy to see that $|N_G(v_1) \cap N_G(v_2)| = 3$. Suppose that $N_G(v_1) = \{v_3, v_4, v_5, v_6\}$ and $N_G(v_2) = \{v_3, v_4, v_5, v_7\}$. Noting that $d_6 = d_7 = 4$ and vertices $v_6$ and $v_7$ must share at least one common neighbour, let $v_3 \in N_G(v_6) \cap N_G(v_7)$. As illustrated in Fig.~\ref{fig8}, $G$ necessarily contains $G_2$ as a subgraph, which implies $|E(G)|\leq 13$.
    \begin{figure}[!h]
    \centering
    \includegraphics[width=1.5in]{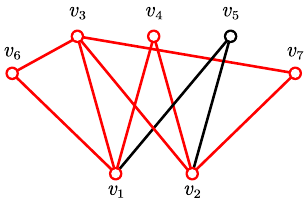}
    \caption{The graph $G$ satisfying $N_G(v_1) \neq N_G(v_2).$}
    \label{fig8}
    \end{figure}

    Further considering $|E(G)| = 13$, we can similarly obtain that the minimum degree $2\leq d_1 \leq 3$. We discuss the two cases separately.

    \textbf{Case 1:} $d_1 = 2$.
    
    Since $d_1=2$, then $|E(G - v_1)| = 13 - 2 = 11 = ex(6, G_2)$. By 1) we have $G - v_1 = G_2'$. Under the condition of being $G_2$-free, the graph $G$ must be obtained from $G_2'$ by adding a vertex of degree $d_1=2$. Consequently, $G=G_2''$ or $G = G_2'''$, both of which are extremal graphs of order $7$ for $G_2$.

    \begin{figure}[!h]
    \centering
    \includegraphics[width=2.5in]{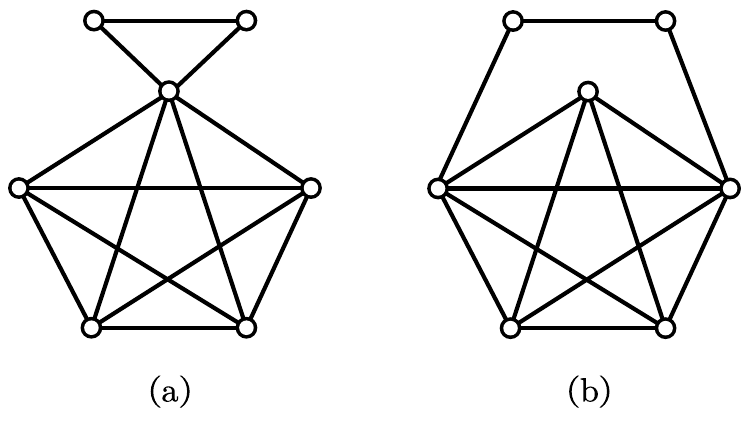}
    \caption{$G''_2$ and $G'''_2$ are 7-order extremal graphs of $G_2$.}
    \end{figure}

    \textbf{Case 2:} $d_1 = 3$.

    Then the degree sequence of $G$ is $\pi_G = (3,3,4,4,4,4,4)$ or $\pi_G = (3,3,3,3,4,5,5)$. 

    \begin{figure}[!h]
    \centering
    \includegraphics[width=3in]{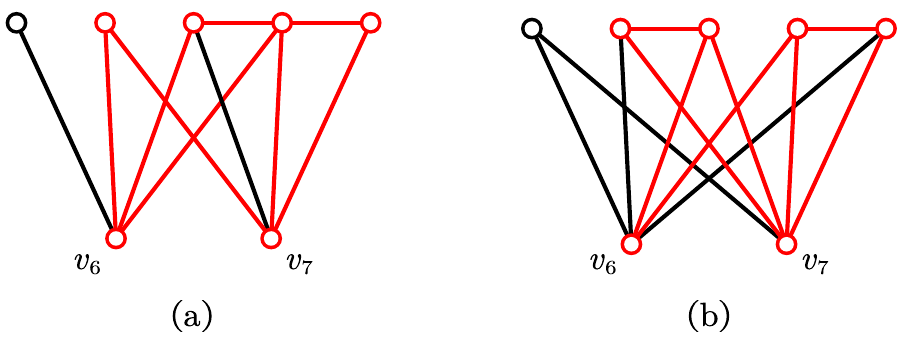}
    \caption{Figure (a) depicts the graph $G$ satisfying $|N_G(v_6) \cap N_G(v_7)| = 3$ in Subcase 2.1; Figure (b) depicts the graph $G$ satisfying $v_6v_7 \notin E(G)$ in Subcase 2.2.}
    \label{fig_add}
    \end{figure}
    
    \textbf{Subcase 2.1:} $\pi_G = (3,3,4,4,4,4,4)$.

    In this subcase, $G$ must have two non-adjacent vertices of degree 4, say $v_6$ and $v_7$. If $|N_G(v_6)\cap N_G(v_7)|=3$, then $G$ necessarily contains a copy of $G_2$ as illustrated in Fig.~\ref{fig_add} (a). Otherwise, if $|N_G(v_6) \cap N_G(v_7)| = 4$, then $G$ must be isomorphic to the graph $(K_2+\overline{K_2}) \vee \overline{K_3} $, which is $G_2$-free. 

    \textbf{Subcase 2.2:} $\pi_G = (3,3,3,3,4,5,5)$.

    In this subcase, $G$ must have two vertices with degree 5, say $v_6$ and $v_7$. If $v_6v_7\notin E(G)$, then $G$ necessarily contains a copy of $G_2$ as illustrated in Fig.~\ref{fig_add} (b). Otherwise, if $v_6v_7\in E(G)$, then $G$ must be isomorphic to the graph $(K_2+K_1) \vee \overline{K_4}$, which is $G_2$-free. 
    
    In summary, we have proved that $|E(G)| \leq 13$, and $|E(G)| = 13$ if and only if $G = (K_2 + \overline{K}_{2} ) \vee \overline{K}_{3}$, $G = (K_1 + \overline{K}_{2} ) \vee \overline{K}_{4}$, $G  = G_2''$, $G = G_2'''$. 

    For 3), suppose $G$ is a graph of order $8$ that contains no $G_2$ as a subgraph.
    
    We first show that $|E(G)| \leq 17$. Assume for contradiction that $|E(G)| = 18$. Then the minimum degree of $G$, denoted by $d_1$, satisfies $d_1\leq 4$. Indeed, if $d_1 \geq 5$, then $|E(G)| \geq \frac{5 \times 8}{2} = 20$, which is a contradiction; consider the graph $G-v_1$ satisfying $|E(G-v_1)| \geq 18 - 4 = 14 > ex(7, G_2)$, which contradicts the fact that $G$ is $G_2$-free. Hence $|E(G)| \leq 17$ and $d_1=4$.

    By 2), we can know $G-v_1 = (K_2 + \overline{K_2}) \vee \overline{K_3}$, $(K_2 + K_1) \vee \overline{K_3}$,  $G_2''$, or $G_2'''$. Under the condition of being $G_2$-free, the graph $G$ must be obtained from the 7-order extremal graphs by adding a vertex of degree $d_1=4$. Consequently, $G = (K_2 +\overline{K_2}) \vee \overline{K_4}$ is the only extremal graph of order $8$ for $G_2$.

\end{IEEEproof}
\subsection{Proof of Lemma~\ref{lem4}}

\GfourLemma*

\begin{IEEEproof}
    For 1), where $n = 6$, it is easy to check that the graph $G$ obtained by deleting any $2$ edges from $K_6$ is not $G_3$-free. Hence $ex(6, G_3) \leq 12$. On the other hand, if we remove any $3$ edges from $K_6$, only $K_3\vee\overline{K_3}$ is $G_3$-free, so $ex(6, G_3) \geq 12$. This proves that $ex(6, G_3) = 12$, and $|E(G)| = 12$ if and only if $K_3\vee\overline{K_3}$ .    

    For 2), suppose that $G$ is a $G_3$-free graph with $7$ vertices and 16 edges.
    
    First we show that the minimum degree of $G$, denoted by $d_1$, must be $4$. If $d_1 \geq 5$, then $|E(G)| \geq \frac{5 \times 7}{2} > 17$, which is a contradiction. If $d_1 \leq 3$, consider the graph $G - v_1$, which satisfies $|E(G - v_1)| \geq 16 - 3 = 13$, which is a contradiction. Thus $d_1 = 4$.

    Next we prove that such a graph $G$ must contain $G_3$ as a subgraph. Consider the graph $G - v_1$, which satisfies $|E(G - v_1)| = 16 - 4 = 12$. By 1), we obtain $G - v_1 = K_3 \vee \overline{K_3}$. The graph $G$ must be obtained from $K_3 \vee \overline{K_3}$ by adding a vertex of degree $4$. Consequently, the graph $G$ must be isomorphic to the graph shown in Fig.~\ref{fig9}, which contains $G_3$ as a subgraph, which implies $|E(G)|\leq15$.

    \begin{figure}[!h]
    \centering
    \includegraphics[width=1.5in]{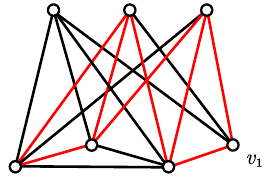}
    \caption{The graph $G$ satisfying $d_1=4$ and $G - v_1 = K_3 \vee \overline{K_3}$.}
    \label{fig9}
    \end{figure}
    
    Further considering $|E(G)| = 15$, we can similarly obtain that the minimum degree $3\leq d_1 \leq 4$. We discuss the two cases separately.

    \textbf{Case 1:} $d_1 = 3$.
    
    Since $d_1=3$, then $|E(G - v_1)| = 15 - 3 = 12 = ex(6, G_3)$. By 1) we have $G - v_1 = K_3 \vee \overline{K_3}$. Under the condition of being $G_3$-free, the graph $G$ must be obtained from $K_3 \vee \overline{K_3}$ by adding a vertex of degree $d_1=3$. Consequently, $G = (K_3 + K_1) \vee \overline{K_3}$ or $G = K_3 \vee \overline{K_4}$, both of which are extremal graphs of order $7$ for $G_3$.

    \textbf{Case 2:} $d_1 = 4$.

    Then the degree sequence of $G$ is either $\pi_G = (4,4,4,4,4,4,6)$ or $\pi_G = (4,4,4,4,4,5,5)$. We consider the following two subcases.

    \textbf{Subcase 2.1:} $\pi_G = (4, 4, 4, 4, 4, 4, 6)$.

    It is easy to see that $G[V \setminus \{v_7\}]$ is a $3$-regular graph of order 6. By Lemma~\ref{lem3}, $G[V \setminus \{v_7\}]$ is isomorphic either to $K_{3,3}$ or to $H_3$. Since $G$ contains no $G_3$ as a subgraph, we must have $G[V \setminus \{v_7\}] = K_{3,3}$, and consequently $G = S_4 \vee \overline{K_3}$, which is also an extremal graph of order $7$ for $G_3$.

    \begin{figure}[!h]
    \centering
    \includegraphics[width=1.5in]{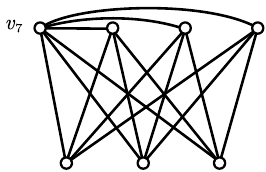}
    \caption{$G = S_4 \vee \overline{K_3}$ is the 7-order extremal graph of $G_3$.}
    \label{fig_sim}
    \end{figure}

    \textbf{Subcase 2.2:} $\pi_G = (4, 4, 4, 4, 4, 5, 5)$.

    In this subcase, $G$ must have two non‑adjacent vertices of degree $4$, say $v_1$ and $v_2$. If $|N_G(v_1) \cap N_G(v_2)| = 4$, then $G$ must be isomorphic to the graph Fig.~\ref{fig_11} (a), which is not $G_3$-free. Otherwise, if $|N_G(v_1) \cap N_G(v_2)| = 3$, then $G$ necessarily contains a copy of $G_3$ as illustrated in Fig.~\ref{fig_11} (b). Therefore, under this degree sequence there exists no graph $G$ satisfying the property of $G_3$-free. 
     \begin{figure}[!h]
    \centering
    \includegraphics[width=3.2in]{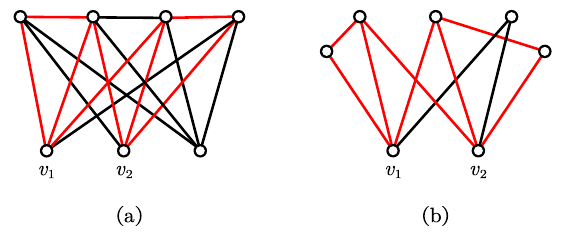}
    \caption{Figures (a) and (b) depict the graph $G$ satisfying $|N_G(v_1) \cap N_G(v_2)| = 4$ and $|N_G(v_1) \cap N_G(v_2)| = 3$ respectively.}
    \label{fig_11}
    \end{figure}

    In summary, we have proved that $|E(G)| \leq 15$, and $|E(G)| = 15$ if and only if $G = (K_3 + K_1) \vee \overline{K_3}$, $G = K_3 \vee \overline{K_4}$, or $G = S_4 \vee \overline{K_3}$.  

    For 3), suppose $G$ is a graph of order $8$ that contains no $G_3$ as a subgraph.
    
    We first show that $|E(G)| \leq 19$. Assume for contradiction that $|E(G)| = 20$. Then the minimum degree of $G$, denoted by $d_1$, must be $5$. Indeed, if $d_1 \geq 6$, then $|E(G)| \geq \frac{6 \times 8}{2} = 24$, which is a contradiction; if $d_1 \leq 4$, consider the graph $G-v_1$ satisfying $|E(G-v_1)| \geq 20 - 4 = 16 > ex(7, G_3)$, which contradicts the fact that $G$ is $G_3$-free. Hence $d_1 = 5$.

    Now consider the graph $G-v_1$. Since $|E(G-v_1)| = 20 - 5 = 15 = ex(7, G_3)$, by 2) we have $G-v_1 = (K_3 + K_1) \vee \overline{K_3}$ or $K_3 \vee \overline{K_4}$, or $S_4 \vee \overline{K_3}$. Note that the minimum degree of $G$ is $d_1 = 5$ while $G-v_1 = (K_3 + K_1) \vee \overline{K_3}$ or $K_3 \vee \overline{K_4}$, or $S_4 \vee \overline{K_3}$, which is a contradiction. Thus $|E(G)| \leq 19$.
    
    However, taking into account the minimum degree of $G$ and the degree sequences of these three graphs, one checks that $G-v_1$ cannot coincide with any of these extremal graphs of order $7$ for $G_3$. This yields a contradiction.

    Now suppose $|E(G)| = 19$. Similarly we can know that its minimum degree must be $d_1 = 4$, and $G-v_1 = (K_3 + K_1) \vee \overline{K_3}$, $K_3 \vee \overline{K_4}$, or $S_4 \vee \overline{K_3}$. Under the condition of being $G_3$-free, the graph $G$ must be obtained from the 8-order extremal graphs by adding a vertex of degree $d_1=4$. Consequently, $G = (K_3 + K_1) \vee \overline{K_4}$ or $G = S_4 \vee \overline{K_4}$. Both graphs have $19$ edges and are $G_3$-free; hence they are extremal graphs of order $8$ for $G_3$.
\end{IEEEproof}

\section*{Acknowledgment}

The authors would like to thank the editors and reviewers for their valuable comments and suggestions.

\bibliographystyle{IEEEtranN}
\bibliography{main}

@article{sfxl,
  author = { Guohua Zhang and  Jiayu Yu and  Yuan Zhou },
  title = { (4, L) Regular Girth-8 {QC-LDPC} Short Codes Based on Reversal of Dual Sequences},
  journal = {Mobile Communications},
  volume = {48},
  number = {05},
  pages = {27-31},
  year = {2024},
  issn = {1006-1010},
}

@ARTICLE{GCD,
  author={Zhang, Guohua and Hu, Yulin and Fang, Yi and Ren, Defeng},
  journal={IEEE Communications Letters}, 
  title={Relation Between {GCD} Constraint and Full-Length Row-Multiplier {QC-LDPC} Codes With Girth Eight}, 
  year={2021},
  volume={25},
  number={9},
  pages={2820-2823},
  doi={10.1109/LCOMM.2021.3096386}}

@ARTICLE{PEG,
  author={Xiaoyu Hu and Eleftheriou, E. and Arnold, D.M.},
  journal={IEEE Transactions on Information Theory}, 
  title={Regular and irregular progressive edge-growth tanner graphs}, 
  year={2005},
  volume={51},
  number={1},
  pages={386-398},
  doi={10.1109/TIT.2004.839541}}

@article{simpleform,
  author    = {Singh, Jasvinder and Gupta, Manish and Bhullar, Jaskarn Singh},
  title     = {On the search of smallest {QC-LDPC} code with girth six and eight},
  journal   = {Cryptography and Communications},
  year      = {2019},
  volume    = {12},
  number    = {},
  pages     = {711-723},
  doi       = {10.1007/s12095-019-00405-2},
}

@ARTICLE{Xiong2,
  author={Xiong, Haoran and Wang, Guanghui and Ma, Zhiming and Yan, Guiying},
  journal={IEEE Transactions on Information Theory}, 
  title={The Impact of the Distance Between Cycles on Elementary Trapping Sets}, 
  year={2026},
  volume={72},
  number={1},
  pages={298-314},
  doi={10.1109/TIT.2025.3631564}}

@electronic{nauty,
  author        = "Brendan McKay",
  title         = "nauty User’s Guide",
  url           = {https://users.cecs.anu.edu.au/~bdm/nauty/},
  year          = "2009",
  key           = "IEEE"
}

@ARTICLE{GallagerLDPC,
  author={Gallager, R.},
  journal={IRE Transactions on Information Theory}, 
  title={Low-density parity-check codes}, 
  year={1962},
  volume={8},
  number={1},
  pages={21-28},
  doi={10.1109/TIT.1962.1057683}}

@ARTICLE{Kari,
  author={Karimi, Mehdi and Banihashemi, Amir H.},
  journal={IEEE Transactions on Information Theory}, 
  title={On Characterization of Elementary Trapping Sets of Variable-Regular {LDPC} Codes}, 
  year={2014},
  volume={60},
  number={9},
  pages={5188-5203},
  doi={10.1109/TIT.2014.2334657}}

@ARTICLE{McGregor,
author = {McGregor, Andrew and Milenkovic, Olgica},
title = {On the hardness of approximating stopping and trapping sets},
year = {2010},
publisher = {IEEE Press},
volume = {56},
number = {4},
issn = {0018-9448},
doi = {10.1109/TIT.2010.2040941},
journal = {IEEE Transactions on Information Theory},
pages = {1640-1650},
numpages = {11}
}

@ARTICLE{TaoXiongfe,
  author={Tao, Xiongfei and Li, Yufei and Liu, Yonghe and Hu, Zuoqi},
  journal={IEEE Communications Letters}, 
  title={On the Construction of {LDPC} Codes Free of Small Trapping Sets by Controlling Cycles}, 
  year={2018},
  volume={22},
  number={1},
  pages={9-12},
  doi={10.1109/LCOMM.2017.2679707}}

@ARTICLE{NaseriBaniha,
  author={Naseri, Sima and Banihashemi, Amir H.},
  journal={IEEE Transactions on Communications}, 
  title={Construction of Time Invariant Spatially Coupled {LDPC} Codes Free of Small Trapping Sets}, 
  year={2021},
  volume={69},
  number={6},
  pages={3485-3501},
  doi={10.1109/TCOMM.2021.3060797}}

@ARTICLE{ChunLDPC,
  author={Sae-Young Chung and Forney, G.D. and Richardson, T.J. and Urbanke, R.},
  journal={IEEE Communications Letters}, 
  title={On the design of low-density parity-check codes within 0.0045 {dB} of the {S}hannon limit}, 
  year={2001},
  volume={5},
  number={2},
  pages={58-60},
  doi={10.1109/4234.905935}}

@INPROCEEDINGS{ChilaER,
  author={Chilappagari, Shashi Kiran and Sankaranarayanan, Sundararajan and Vasic, Bane},
  booktitle={2006 IEEE International Conference on Communications}, 
  title={Error Floors of {LDPC} Codes on the Binary Symmetric Channel}, 
  year={2006},
  volume={3},
  number={},
  pages={1089-1094},
  doi={10.1109/ICC.2006.254892}}

@article{cole2006,
    title={A General Method for Finding Low Error Rates of {LDPC} Codes}, 
    author={Chad A. Cole and Stephen G. Wilson and Eric. K. Hall and Thomas R. Giallorenzi},
    year={2006},
    journal={arXiv preprint cs/0605051},
    eprint={cs/0605051},
    archivePrefix={arXiv},
    primaryClass={cs.IT},
}

@ARTICLE{Milenkovic,
  author={Milenkovic, Olgica and Soljanin, Emina and Whiting, Philip},
  journal={IEEE Transactions on Information Theory}, 
  title={Asymptotic Spectra of Trapping Sets in Regular and Irregular {LDPC} Code Ensembles}, 
  year={2007},
  volume={53},
  number={1},
  pages={39-55},
  doi={10.1109/TIT.2006.887060}}

@ARTICLE{Butler1,
  author={Butler, Brian K. and Siegel, Paul H.},
  journal={IEEE Transactions on Information Theory}, 
  title={Error Floor Approximation for {LDPC} Codes in the {AWGN} Channel}, 
  year={2014},
  volume={60},
  number={12},
  pages={7416-7441},
  doi={10.1109/TIT.2014.2363832}}

@ARTICLE{HashemAmir,
  author={Hashemi, Yoones and Banihashemi, Amir H.},
  journal={IEEE Communications Letters}, 
  title={Lower Bounds on the Size of Smallest Elementary and Non-Elementary Trapping Sets in Variable-Regular {LDPC} Codes}, 
  year={2017},
  volume={21},
  number={9},
  pages={1905-1908},
  doi={10.1109/LCOMM.2017.2707553}}

@INPROCEEDINGS{XiongYeYan,
  author={Xiong, Haoran and Ye, Zicheng and Zhang, Huazi and Wang, Jun and Liu, Ke and Yin, Dawei and Wang, Guanghui and Yan, Guiying and Ma, Zhiming},
  booktitle={2024 IEEE Information Theory Workshop (ITW)}, 
  title={Theoretical Bounds for the Size of Elementary Trapping Sets by Graph Theory Methods}, 
  year={2024},
  volume={},
  number={},
  pages={193-198},
  doi={10.1109/ITW61385.2024.10806922}}

@ARTICLE{theta,
  author={Zhai, M and Fang, L and Shu,J},
  journal={Graphs and Combinatorics}, 
  title={On the {T}ur\'{a}n Number of Theta Graphs}, 
  year={2021},
  volume={37},
  page={2155–2165},
}

@ARTICLE{turan,
  author={Tur\'{a}n, P},
  journal={Matematikai e\'{s} Fizikai Lapok}, 
  title={On an Extremal Problem in Graph Theory}, 
  year={1941},
  volume={48},
  pages={436-452},
}

@INPROCEEDINGS{KoetterVN,
  author={Koetter, R. and Li, W.-C.W. and Vontobel, P.O. and Walker, J.L.},
  booktitle={Information Theory Workshop}, 
  title={Pseudo-codewords of cycle codes via zeta functions}, 
  year={2004},
  volume={},
  number={},
  pages={7-12},
  doi={10.1109/ITW.2004.1405265}}

@ARTICLE{Tanner1,
  author={Tanner, R.},
  journal={IEEE Transactions on Information Theory}, 
  title={A recursive approach to low complexity codes}, 
  year={1981},
  volume={27},
  number={5},
  pages={533-547},
  doi={10.1109/TIT.1981.1056404}}

@INPROCEEDINGS{AmirzadePho,
  author={Amirzade, Farzane and Sadeghi, Mohammad-Reza and Panario, Daniel},
  booktitle={2022 17th Canadian Workshop on Information Theory (CWIT)}, 
  title={Protograph-based {LDPC} codes with chordless short cycles and large minimum distance}, 
  year={2022},
  volume={},
  number={},
  pages={16-20},
  doi={10.1109/CWIT55308.2022.9817675}}

@ARTICLE{Amirzade,
  author={Amirzade, Farzane and Sadeghi, Mohammad-Reza and Panario, Daniel},
  journal={IEEE Transactions on Information Theory}, 
  title={Construction of Protograph-Based {LDPC} Codes With Chordless Short Cycles}, 
  year={2024},
  volume={70},
  number={1},
  pages={51-74},
  doi={10.1109/TIT.2023.3307583}}

@ARTICLE{Sadeghi,
  author={Sadeghi, Mohammad-Reza and Amirzade, Farzane},
  journal={IEEE Communications Letters}, 
  title={Edge-Coloring Technique to Analyze Elementary Trapping Sets of Spatially-Coupled {LDPC} Convolutional Codes}, 
  year={2020},
  volume={24},
  number={4},
  pages={711-715},
  doi={10.1109/LCOMM.2019.2962671}}

@ARTICLE{NaseriSiAm,
  author={Naseri, Sima and Banihashemi, Amir H.},
  journal={IEEE Communications Letters}, 
  title={Construction of Girth-8 {QC-LDPC} Codes Free of Small Trapping Sets}, 
  year={2019},
  volume={23},
  number={11},
  pages={1904-1908},
  doi={10.1109/LCOMM.2019.2933828}}

\end{document}